\documentclass[letterpaper]{article}
\usepackage{preamble} \usepackage{authblk}
\usepackage[ruled,vlined,linesnumbered]{algorithm2e}

\newcommand{\dep}{\operatorname{depth}}
\newcommand{\pare}{\operatorname{par}}

\title{A Simplified Analysis of the Good-Bad $3/2$-Approximation Algorithm for Some Minimum-Cost Graph Problems}
\author{Shayan Ranjbarzadeh}
\author{David P.\ Williamson}
\author{Hannane Yaghoubizade}
\affil{Cornell University, School of Operations Research and Information Engineering}
\date{}

\begin{document}
\maketitle

\begin{abstract}
In this paper, we consider an easy greedy approximation algorithm, the {\em good-bad algorithm}, introduced by \textcite{Couetoux11} for finding a minimum-cost set of edges such that every connected component has at least $k$ vertices. Cou\"etoux proves that the good-bad algorithm achieves a $3/2$-approximation for this problem. \textcite{DavisW12} extend this result to the more general problem of finding a minimum-cost edge set that contains at least one edge from every cut $S\subseteq V$ satisfying $h(S) = 1$ where $h:2^V \rightarrow \{0,1\}$ is downward monotone; that is, $h(S) = 1$ implies $h(T) = 1$ for every nonempty subset $T \subseteq S$. The original problem corresponds to $h(S) =1$ when $|S|<k$.  We give a simplified analysis of the good-bad algorithm for downward monotone functions.  
\end{abstract}

\section{Introduction}

Consider the following problem, which we will refer to as \textsc{Lower Capacitated Tree Partitioning}. Given an undirected graph $G=(V,E)$ with nonnegative edge costs $c(e) \geq 0$ for all $e \in E$, and a positive integer $k$, the goal is to find a minimum-cost edge set $F\subseteq E$ such that each connected component of $(V,F)$ has at least $k$ vertices.  If $k=2$, the problem is min-cost edge cover, and if $k=|V|$, it is minimum spanning tree; both of these problems are known to be polynomial-time solvable.  
However, for any other constant value of $k$, \textsc{Lower Capacitated Tree Partitioning} is known to be NP-hard: \textcite{ImielinskaKK93} establish this for $k\geq 4$, while \textcite{BazganCT11} prove it for $k=3$ and, more generally, show that the problem is APX-hard for every constant $k\geq 3$.

\textcite{ImielinskaKK93} give a \(2\)-approximation algorithm for this problem through a natural adaptation of Kruskal's algorithm \cite{Kruskal56}. As in Kruskal's algorithm, edges are considered in nondecreasing order of cost; however, an edge between two distinct components is selected only when at least one of the components has size less than $k$. The performance guarantee of the algorithm is tight.  Other 2-approximation algorithms based on edge deletion (instead of, or in addition to, edge insertion) are given by \textcite{LaszloM05,LaszloM06}, who also perform
some experimental comparisons.  

\textcite{GoemansW94} show that the algorithm of Imieli\'nska et al.\ can be extended to obtain a \(2\)-approximation algorithm for a broad class of graph problems characterized by a \emph{downward monotone} function \(h \colon 2^V \to \{0,1\}\)\footnote{In fact, \textcite{GoemansW94} consider the more general setting of integer-valued downward monotone functions \(h \colon 2^V \to \mathbb{Z}_{\geq 0}\), whereas we restrict our attention to binary-valued functions.}. Specifically, \(h\) is downward monotone if \(h(S)=1\) implies \(h(T)=1\) for every nonempty subset \(T \subseteq S\). A set of edges \(F\) is feasible for the problem defined by \(h\) if
\(
    \lvert F \cap \delta(S) \rvert \geq h(S)
\)
for every nontrivial subset \(S \subseteq V\), where \(\delta(S)\) is the set of edges with exactly one endpoint in \(S\). For example, the \textsc{Lower Capacitated Tree Partitioning} corresponds to the function \(h\) where \(h(S)=1\) if and only if \(\lvert S\rvert<k\).
For this larger class of problems, similar to before, the algorithm processes the edges in nondecreasing order of cost, and adds an edge whenever it connects two distinct components \(C\) and \(C'\) such that \(h(C)=1\) or \(h(C')=1\). The edge-deletion algorithm of \textcite{LaszloM05,LaszloM06} also provides a 2-approximation guarantee for these problems \cite{LaszloM08}. \textcite{GoemansW94} show that several graph problems can be modeled with downward monotone functions, including some location-design and location-routing problems; for example, they consider a problem in which every connected component not only must have at least $k$ vertices, but also must have an open depot from a subset $D \subseteq V$, where there is a cost $c(d)$ for opening the depot $d \in D$ to serve the component.  

\textcite{Couetoux11} gives a $3/2$-approximation algorithm for \textsc{Lower Capacitated Tree Partitioning}.
In any partial solution constructed by the algorithm, let us call a connected component {\em small} if it has fewer than $k$ vertices in it, and {\em big} otherwise.  We call an edge {\em good} if it joins two small components into a big component, and {\em bad} otherwise.  
The algorithm behaves similarly to the Imieli\'nska et al.\ algorithm in that it considers adding edges to the solution in nondecreasing order of cost.  
However, when it considers adding an edge $e$ to the solution at cost $c(e)$, it also considers all good edges $e'$ of cost $c(e') \leq 2c(e)$.  
If such a good edge exists, it adds the cheapest one; otherwise, it adds the edge $e$.  
The intuition for the algorithm is that we want to try to decrease the number of small components of a partial solution until there is no small components left. A good edge decreases this number by two, whereas a bad edge decreases it by only one. It is therefore reasonable to pay up to twice as much for a good edge as for a bad edge. We will refer to this algorithm as the {\em good-bad algorithm.} \textcite{DavisW12} show that the natural extension of the good-bad algorithm to downward monotone functions, in which a component $C$ is small if $h(C)=1$ and big if $h(C)=0$ also obtains a 3/2-approximation guarantee.\footnote{\textcite{CouetouxDW13} wrote a joint journal version of these results.}

In this paper, we provide a simpler analysis of the good-bad algorithm based on a straightforward charging argument. Cou\"etoux's analysis relies on a 
more intricate charging scheme, whereas Davis and Williamson use a primal-dual framework that compares the algorithm's cost directly with the optimal solution rather than with the value of a dual solution. We believe that our analysis is considerably simpler than these previous approaches. Moreover, its clean structure allows us to derive tighter guarantees in certain cases.

In the next section, we give a formal presentation of the good-bad algorithm, and in Section \ref{sec:analysis}, we show that it obtains a 3/2-approximation guarantee.

\section{Preliminaries}
For an undirected graph \(G=(V,E)\), a \emph{cut} is a nonempty proper subset \(S\subsetneq V\). For a subset of edges $F \subseteq E$, we write \(\delta_F(S)\) for the set of edges in $F$ crossing the cut $S$, namely, the edges in $F$ with exactly one endpoint in \(S\). When the edge set is omitted, we write \(\delta(S)\) as shorthand for \(\delta_E(S)\). A set function $h:2^V\to \{0,1\}$ is called \emph{downward monotone} if, whenever $h(S) = 1$ for some $S\subseteq V$, it also holds that $h(T) =1$ for every $T\subseteq S$. 
Given an undirected graph $G=(V, E)$, nonnegative edge-costs \(c \colon E \to \mathbb{R}_{\geq 0}\) and a prespecified downward monotone function \(h \colon 2^V \to \{0,1\}\), our goal is to (approximately) solve the following optimization problem
\begin{align}\label{eq:problem}
    \min_{F \subseteq E}
    \left\{
        \sum_{e \in F} c(e)
        \;\middle|\;
       \lvert\delta_F(S)\rvert \geq h(S)
        \text{ for every cut } S \subsetneq V, S\ne \emptyset
    \right\}.
\end{align}
    
\subsection{The Algorithm}
For a partial solution $F\subseteq E$, let us call a connected component $C$ of $G_F=(V, F)$ \emph{small} if $h(C) = 1$, and \emph{big} if $h(C) = 0$. 
For an edge \(e\in E\), let \(\Delta_F(e)\) denote the decrease in the number of small connected components resulting from adding \(e\) to \(F\).

An edge $e\in E$ is \emph{available} with respect to $F$, if $\Delta_F(e) \geq 1$, or equivalently, if its two endpoints belong to different connected components, at least one of which is small. 
An available edge \(e\) is called \emph{good} if \(\Delta_F(e)=2\), or equivalently, if it merges two small components into a big component. Otherwise, \(e\) is called \emph{bad}. A bad edge therefore either merges two small components into another small component or joins a small component and a big component, resulting in a big component.\footnote{By the downward monotonicity of \(h\), any set containing a big component must itself be big.} In either case, \(\Delta_F(e)=1\). Consequently,
\begin{align*}
    \Delta_F(e)= 
    \begin{cases}
        2 & \text{$e$ is good}, \\
        1 & \text{$e$ is bad}, \\
        0 & \text{$e$ is not available}.
    \end{cases}
\end{align*}

The good-bad algorithm starts from $F = \emptyset$, and iteratively adds available edges until no small connected components remain. At each iteration, given the current partial solution $F$, the algorithm selects an available edge $e\in E$ that minimizes the cost-to-benefit ratio ${c(e)}/{\Delta_F(e)}$. 
In other words, the algorithm selects the cheapest good edge \(e\) whenever ${c(e)}/{2} \le c(e')$ for every bad edge $e'$; otherwise, it selects the cheapest available bad edge. \cref{alg:good-bad} is a formal description of this procedure.

\begin{algorithm}[t]
\caption{Good-Bad Algorithm}
\label{alg:good-bad}

\KwIn{Graph $G = (V, E)$ along with edge costs $c:E \to \mathbb{R}_{\geq 0}$.} 
\KwOut{A subset of edges $F$ such that every connected component of \(G_F=(V,F)\) is big.}

$F \gets \emptyset$\\
$S \gets \left\{\{v\} \mid  v\in V, h(\{v\}) = 1\right\}$, the set of small components of $G_F=(V, \emptyset).$

\While{$|S| > 0$}{
    $e^*\gets \arg\min\left\{{c(e)}/{\Delta_F(e)}\mid e\in E, \Delta_F(e) > 0\right\}$, an available edge minimizing cost-to-benefit ratio.\\
    $F \gets F \cup \{e^*\}$\\
    $S \gets$ the set of small components of $G_F=(V, F)$.
}

\Return{$F$}
\end{algorithm}

We will denote the final output of good-bad algorithm by $\textsc{Alg} \subseteq E$. Note that \(\textsc{Alg} \) is feasible for our problem. To see this, consider any cut \(S \subsetneq V\). If $|\delta_\textsc{Alg} (S)| \geq 1$, then $|\delta_\textsc{Alg} (S)| \geq 1 \geq h(S),$
since \(h(S) \in \{0,1\}\). Otherwise, \(\delta_\textsc{Alg} (S)=\emptyset\), so \(S\) is a union of connected components of \((V,\textsc{Alg} )\). At termination, every connected component \(C\) is big and therefore satisfies \(h(C)=0\). Since \(C \subseteq S\), downward monotonicity implies that \(h(S)=0\). Hence, \(\lvert\delta_\textsc{Alg} (S) \rvert = 0 = h(S).\)
Thus, every cut requirement is satisfied, and \(\textsc{Alg} \) is feasible.

\section{Analysis of the Algorithm}
\label{sec:analysis}

Consider an optimal solution and let $T^*_1, \dots, T^*_k$ denote its connected components in an arbitrary order. Note that as $h(S) \in \{0,1\}$, each connected component is a tree.\footnote{Since the edge costs are nonnegative, we might have cycles including 0-cost edges, but such edges can be ignored w.l.o.g.} For each component $T^*_i$, let $e^*_i\in E(T^*_i)$ be an edge with maximum cost in $E(T^*_i)$, that is $e^*_i = \arg\max_{e\in E(T^*_i)} c(e)$. 
Writing \(e_i^*=\{u_i,v_i\}\), we view \(T_i^*\) as birooted at
\(u_i\) and \(v_i\), and denote the set of its roots by
\(R_i^*=\{u_i,v_i\}\). Birooted components were also used by
\textcite{DavisW12} in their analysis.
For each vertex $v \in V(T^*_i)$, we define its depth, denoted by $\dep(v)$, as its distance to the nearest root in $R^*_i$. For every $v \in V(T^*_i)\setminus R^*_i$ we use $\pare(v)$ to denote its parent in birooted $T^*_i$.

We next define a strict total order \(\prec\) on \(V\). We first order vertices according to the index of the component to which they belong and then, within each component, according to their depth. Specifically, for vertices
\(v \in V(T_i^*)\) and \(u \in V(T_j^*)\),
\(v \prec u\) if \(i<j\), or if \(i=j\) and
\(\dep(v)<\dep(u)\). Vertices belonging to the same component and having
the same depth are ordered according to an arbitrary fixed tie-breaking
rule.

For the analysis, we inductively maintain a unique \emph{representative}
for each small component of \(G_F=(V,F)\), while every other vertex has
already been \emph{eliminated} and assigned a price.
Initially, each vertex \(v\) with \(h(\{v\})=1\) is designated as the 
representative of the singleton component \(\{v\}\). Each vertex \(v\) with 
\(h(\{v\})=0\) is eliminated and assigned price \(p(v)=0\). Whenever the 
algorithm adds an edge \(e\) to \(F\), we update the representatives and 
prices according to the following scheme.
\begin{itemize}
    \item If \(e\) is a \emph{good} edge; let $C_1$ and $C_2$ be the small components $e$ connects and \(v_1\in V(C_1)\) and \(v_2\in V(C_2)\) be their unique representatives, respectively. We eliminate both representatives and assign each a price of \(c(e)/2\); that is, \( p({v_1})=p({v_2})={c(e)}/{2}\).
    In this case, merging \(C_1\) and \(C_2\) produces a big component, which has no representative.
    \item If $e$ is a \emph{bad} edge connecting a small component $C$ with the unique representative $v\in V(C)$ and a big component with no representative; we eliminate $v$ and assign it a price of $c(e)$, that is, $p(v) = c(e)$. Thus, the resulting big component has no representative.

    \item If $e$ is a \emph{bad} edge connecting two small components $C_1$ and $C_2$ with representatives $v_1\in V(C_1)$ and $v_2\in V(C_2)$; we eliminate the representative among $\{v_1, v_2\}$ that appears later in the order \(\prec\), and assign it a price of $c(e)$. The other vertex then will be the representative of the newly merged small component. For instance, if $v_2 \prec v_1$, we let $p(v_1) = c(e)$, eliminate $v_1$, and let $v_2$ be the represantative of the new component.
\end{itemize}
Observe that, upon termination of the algorithm, as every connected component is big, every vertex has been eliminated and assigned a price. Furthermore, whenever an edge \(e\) is added to the solution, a total price of \(c(e)\) is assigned to one or two representatives, which are then eliminated. Since an eliminated representative is never assigned another price, we will have  \(\sum_{e\in \textsc{Alg} } c(e)=\sum_{v\in V} p(v)\) at the end.

\begin{lemma} \label{representative}
    At every iteration of the algorithm, the representative of each small connected component \(C\) is the vertex $v\in V(C)$ satisfying \(v \prec u\) for all \(u \in V(C)\setminus\{v\}\).
\end{lemma}

\begin{proof}
    The proof proceeds by a straightforward induction on the size of the partial solution.
\end{proof}

\begin{lemma} \label{available-edge}
    Suppose that \(v\) is eliminated and assigned a price during an iteration, and let \(F\) denote the partial solution immediately before that iteration. Then, for every edge \(e'\in E\) available at that iteration, it holds that \[p(v) \leq \frac{c(e')}{\Delta_F(e')}.\]
\end{lemma}

\begin{proof}
    In the iteration when $v$ is eliminated, \cref{alg:good-bad} selects an edge $e$ minimizing the cost-to-benefit ratio ${c(e)}/{\Delta_F(e)}$, meaning ${c(e)}/{\Delta_F(e)} \leq {c(e')}/{\Delta_F(e')}$ for every edge $e'$ available at that iteration. From the charging scheme we have $p(v) = c(e)/2$ if $e$ is good,  and $p(v) = c(e)$, if $e$ is bad, meaning $p(v) = {c(e)}/{\Delta_F(e)}$, and completing the proof.
\end{proof}

\begin{lemma} \label{parent-edge}
    For every non-root vertex \(v\in V \setminus \cup_{i\in k} R^*_i\), the edge \(e_v=\{v,\pare(v)\}\) remains available until \(v\) is eliminated in the charging scheme.
\end{lemma}

\begin{proof}
    Note that, an edge remains available whenever its endpoints lie in distinct connected components and at least one of these components is small. Suppose that $v$ is eliminated at iteration $t$. Thus, $v$ is the unique representative of the component containing \(v\) throughout all iterations preceding \(t\), and in particular, this component must always be small. Thus, it suffices to show that the two components containing $v$ and $\pare(v)$ do not merge prior to $t$. 
    
    Assume for the sake of contradiction that the components containing $\pare(v)$ and $v$ merge at some iteration $t' < t$. Let \(C_v\) and \(C_{\pare}\) denote the components containing \(v\) and \(\pare(v)\), respectively, immediately before this merge, and let \(C\) denote the resulting component.
    
    As $v$ is eliminated later at iteration $t$, and it remains to be the representative of its component until iteration $t$, components $C_v$ and $C$ must be small. Since $C$ is small,  $C_{\pare}$ is also small and therefore has a representative $u$. 
    By ~\cref{representative} it holds that $u \prec \pare(v)$. 
    On the other hand, because $\pare(v)$ and $v$ belong to the same $T^*_i$ and $\dep(\pare(v)) < \dep(v)$,  it holds that $\pare(v) \prec v$. 
    By transitivity, we obtain $u  \prec v$. 
    So, it follows from the charging scheme that \(v\) would have been eliminated at iteration \(t'\), contradicting the assumption that \(v\) is eliminated at iteration \(t>t'\).
\end{proof}

\begin{corollary}\label{tree-price}
    For every vertex \(v\in V \setminus \cup_{i\in k} R^*_i\), it holds that $p(v) \leq c(e_v)$ where $e_v=\{v, \pare(v)\}$.
\end{corollary}
\begin{proof}
    By \cref{parent-edge}, $e_v$ is available right before elimination of $v$, and therefore, by \cref{available-edge} we get $p(v) \leq {c(e_v)}/{\Delta_F(e_v)} \leq c(e_v)$.
\end{proof}

\begin{lemma} \label{root-price}
    For every  optimal tree $T^*_i$ and $v \in V(T^*_i)$, it holds that $p(v) \leq c(e^*_i)$ where $e^*_i$ is the maximum-cost edge in $T^*_i$.\footnote{For non-root vertices $v\in V(T_i^*)\setminus R^*_i$, this can be directly concluded from \cref{tree-price}, but the included proof works for both roots and non-root vertices.}
\end{lemma}
\begin{proof}
    Consider the iteration in which $v$ is eliminated, and let $C_v$ be the connected component containing $v$ before the elimination. Note that $V(C_v)$ cannot contain all of $V(T^*_i)$. Indeed, since \(T_i^*\) is big, downward monotonicity would imply that \(C_v\) is also big, meaning all vertices of $C_v$, including $v$, would already be eliminated. 
    
    Since $v\in V(T^*_i)\cap V(C_v)$,  $V(T^*_i) \setminus V(C_v) \ne \emptyset$, and $T^*_i$ is connected, there exists an edge $e\in \delta(V(C_v))\cap E(T^*_i)$. As $e\in \delta(V(C_v))$, the endpoints of $e$ belong to different connected components at the current iteration, one of which is $C_v$. Because $C_v$ is small, $e$ must be available. By \cref{available-edge} and since $e^*_i$ is a maximum-cost edge, we get $p(v) \leq c(e) \leq c(e^*_i)$.

\end{proof}

\begin{lemma} \label{2-root-price}
    For every optimal component $T^*_i$ with roots $R^*_i = \{u_i,v_i\}$ it holds that \[p(u_i) + p(v_i) \le c(e^*_i) + \frac{c(T^*_i)}{2},\]
where \(c(T^*_i) \coloneqq \sum_{e\in E(T^*_i)} c(e)\).
\end{lemma}

\begin{proof}
    Fix an optimal component \(T^*\in\{T_1^*,\ldots,T_k^*\}\), and let
    \(R^*=\{u, v\}\) denote its root set. Recall that
    \(e^*=\{u, v\}\) is a maximum-cost edge of \(T^*\). Let $T^*_u$ and $T^*_v$ denote the sub-trees of $(V(T^*), E(T^*)\setminus e^*)$ containing $u$ and $v$, respectively.
    
    As \(e^*=\{u, v\}\) is a maximum-cost edge, for every $e\in E(T^*)\setminus \{e^*\}$ we have $c(e) \leq c(T^*)/2$.
    Let \(t\) be the first iteration at which either \(u\) or \(v\) is eliminated, and assume without loss of generality that \(v\) is eliminated at iteration \(t\). Since every connected component has a unique non-eliminated vertex, its representative, and neither \(u\) nor \(v\) has been eliminated before iteration \(t\), the two vertices must belong to distinct small components immediately prior to iteration \(t\). Denote these components by \(C_u\) and \(C_v\), respectively, and let $F$ be the partial solution of the algorithm right before iteration $t$. We consider the following two cases:

    \begin{itemize}
        \item \textbf{Case 1.} There exists some edge $e \in \delta(V(C_u))\cap E(T^*_u)$ or $e \in \delta(V(C_v))\cap E(T^*_v)$ at the beginning of iteration $t$. Endpoints of $e$ belong to different connected components at the current iteration, one of which is $C_v$ or $C_u$, and as $C_v$ and $C_u$ are small, $e$ must be available. By \cref{available-edge} we know $p(v) \leq c(e)$, and by \cref{root-price} we have $p(u) \leq c(e^*)$. Since $e \in E(T^*) \setminus e^*$, we have $c(e) \le {c(T^*)}/{2}$, implying
        \[
            p(u) + p(v) \leq  c(e^*)+c(e) \leq c(e^*)+\frac{c(T^*)}{2}. 
        \]
        \item \textbf{Case 2.} There does not exist any edge $e \in \delta(V(C_u))\cap E(T^*_u)$ or $e \in \delta(V(C_v))\cap E(T^*_v)$. In this case, we must have $C_u \cup C_v \supseteq T^*_u \cup T^*_v$, meaning $C_u \cup C_v$ is big by downward monotonicity. Thus, $e^*$ connects two small connected components $C_v$ and $C_u$ to form a big component, and therefore, it is a good edge. By \cref{available-edge}, we know $p(v) \leq {c(e^*)}/{2}$, and by \cref{root-price} we have $p(u) \leq c(e^*)$ implying
        \[
            p(u) + p(v) \leq  c(e^*) +\frac{c(e^*)}{2} \leq c(e^*) + \frac{c(T^*)}{2}.
        \]\qedhere
    \end{itemize}
\end{proof}

\begin{theorem} \label{thm:main_approx}
    The good-bad algorithm achieves an approximation ratio of \({3}/{2}\).
\end{theorem}

\begin{proof}
    Summing the inequality $p(v) \le c(e_v)$, established in \cref{tree-price} over all non-root vertices of a component $T^*_i$ gives:
    \begin{equation} \label{eq:nonroot_bound}
        \sum_{v \in V(T^*_i) \setminus \{u_i, v_i\}} p(v) \le \sum_{v \in V(T^*_i) \setminus \{u_i, v_i\}} c(e_v) = c(T^*_i) - c(e^*).
    \end{equation}

    For the roots of the component, by Lemma~\ref{2-root-price}, the assigned prices satisfy \(p(u_i) + p(v_i) \le c(e^*_i) + \frac{c(T^*_i)}{2}\).
    Combining with \cref{eq:nonroot_bound} the total price assigned to all vertices in $T^*_i$ is bounded by
    \begin{align*}
        \sum_{v \in V(T^*_i)} p(v) = (p(u_i) + p(v_i)) + \sum_{v \in V(T^*_i) \setminus \{u_i, v_i\}} p(v) 
        \le \left( c(e^*_i) + \frac{c(T^*_i)}{2} \right) + \left( c(T^*_i) - c(e^*_i) \right) 
        = \frac{3}{2}\cdot c(T^*_i).
    \end{align*}

    Summing this across all  components $T_1^*, \dots, T_k^*$, we get $\sum_{v \in V} p(v) \leq \frac{3}{2}\cdot \text{OPT}$. Finally, recall that from the charging scheme we had $\sum_{v \in V} p(v) = \sum_{e\in \textsc{Alg}} c(e)$, completing the proof.
\end{proof}

\begin{proposition} \label{prop:M+1}
    The good-bad algorithm achieves an approximation ratio of $M+1$ where $M \coloneqq \max_{i\in [k]}{c(e^*_i)}/{c(T^*_i)}$.
\end{proposition}

\begin{proof}
    In the proof of \cref{2-root-price}, fixing $T^*\in \{T^*_1, \dots, T^*_k\}$ with the maximum-cost edge $e^*=\{v, u\}$, we bound $p(u) + p(v)$ in two cases:
    \begin{itemize}
        \item \textbf{Case 1.} $p(u) + p(v) \leq c(e^*) + c(e).$
        \item \textbf{Case 2.} $p(u) + p(v) \leq  c(e^*)+c(e^*)/2.$
    \end{itemize}
    Both of these bounds are at most $2c(e^*)=c(e^*) + \left({c(e^*)}/{c(T^*)}\right) \cdot c(T^*) \le c(e^*) +M\cdot c(T^*)$. Combining this new bound with \cref{eq:nonroot_bound} gives us
    \begin{align*}
        \sum_{v \in V(T^*_i)} p(v) &= (p(u_i) + p(v_i)) + \sum_{v \in V(T^*_i) \setminus \{u_i, v_i\}} p(v) \\
        &\le \left( c(e^*_i) + M \cdot c(T^*_i) \right) + \left( c(T^*_i) - c(e^*_i) \right) 
        = \left(M + 1\right)\cdot c(T^*_i).
    \end{align*}
    Summing the above inequality across all components $T_1^*, \dots, T_k^*$ results in an approximation ratio of $M+1$. 
\end{proof}

\begin{corollary}\label{cor:better-bound}
    The good-bad algorithm achieves an approximation ratio of $\min(M+1, 3/2)$ where $M \coloneqq \max_{i\in [k]}{c(e^*_i)}/{c(T^*_i)}$.
\end{corollary}
\begin{proof}
    It is directly implied from \cref{thm:main_approx} and \cref{prop:M+1}.
\end{proof}

We note that by using Corollary \ref{cor:better-bound}, we can derive that for unweighted \textsc{Lower Capacitated Tree Partitioning}, the good-bad algorithm has an approximation ratio of ${k}/{k-1}$, since $M\leq {1}/{k-1}$ for each optimal component $T^*_i$. Indeed, each edge has weight 1, and each component must have at least $k-1$ edges.  However, Bagzan, Cou\"etoux, and Tuza \cite{BazganCT11} have already observed that any spanning tree is a $\frac{k}{k-1}$-approximation to the problem: any feasible solution will have at least $n/k$ connected components, so the optimal solution must have at least $n-n/k = \left((k-1)/k\right)n$ edges, so that for a spanning tree the approximation ratio is at most $k/(k-1)$.

\section{Conclusion}

In this paper, we gave a simple alternative proof that the good-bad algorithm of \textcite{Couetoux11} achieves an approximation ratio of \(3/2\) for downward monotone functions. Many network-design problems admit relatively simple \(2\)-approximation algorithms, but breaking the factor-\(2\) barrier is often difficult, as in Steiner forest \cite{AhmadiGHJM25} and connectivity augmentation \cite{TraubZ23}. For downward monotone functions, by contrast, both the algorithm and the analysis yielding the \(3/2\)-approximation are simple. This raises the question of whether ideas from the good-bad algorithm or our charging analysis could lead to simpler algorithms with approximation guarantees below \(2\) for Steiner forest, connectivity augmentation, or related problems.

\subsection*{Acknowledgments}

AI was used to improve the writing of parts of this paper.  The authors assume responsibility for all content.

\printbibliography
\end{document}